\documentclass[preprint,3p, twocolumn, review]{elsarticle}
\usepackage{ucs}
\usepackage[utf8x]{inputenc}
\usepackage{hyperref}
\usepackage[T1]{fontenc}
\usepackage{listings}
\usepackage{graphicx}
\usepackage{makeidx}
\usepackage{amssymb}
\usepackage{appendix}
\usepackage{pdfpages}
\usepackage{amsmath}
\usepackage{amsthm}
\usepackage{tabularx}
\usepackage{algpseudocode}
\usepackage{algorithm}
\usepackage{mathtools}

\newtheorem{theorem}{Theorem}
\newtheorem{lemma}[theorem]{Lemma}
\newdefinition{definition}[theorem]{Definition}

\newcommand{\F}{{\rm F}} % unary constraint {0}
\newcommand{\T}{{\rm T}} % unary constraint {1}

\newcommand{\R}{\ensuremath{R^{\scriptscriptstyle 1/3}}}

\newcommand{\Rddd}{\ensuremath{R_{\scriptscriptstyle 3
      \neq}^{\scriptscriptstyle 1/3}}}

\newcommand{\orn}[2]{\ensuremath{\mathrm{OR}^{
      #2}_{\scriptscriptstyle #1}}}
\newcommand{\nandn}[2]{\ensuremath{\mathrm{NAND}^{#2}_{\scriptscriptstyle
      #1}}}
\newcommand{\oddn}[2]{\ensuremath{\mathrm{ODD}^{#2}_{\scriptscriptstyle
      #1}}}
\newcommand{\evenn}[2]{\ensuremath{\mathrm{EVEN}^{#2}_{\scriptscriptstyle
      #1}}}

\newcommand{\coclone}[3]{\ensuremath{\mathrm{#1}}^{#2}_{#3}}
\newcommand{\deducesto}[1]{\ensuremath{\xhookrightarrow{#1}}}

\newcommand{\Or}{\ensuremath{\textsc{OR}}}

\newcommand{\overbar}[1]{\mkern 2.7mu\overline{\mkern-2.7mu#1\mkern-2.7mu}\mkern 2.7mu}
\newcommand{\clone}[1]{\ensuremath{\mathcal{#1}}}
\newcommand{\cclone}[1]{\ensuremath{\langle #1 \rangle}}
\newcommand{\pcclone}[1]{\ensuremath{\langle #1 \rangle_{\not \exists}}}

\begin{document}
\begin{frontmatter}
\title{Weak Bases of Boolean Co-Clones}
%\author{Victor Lagerkvist%\inst{1}
%\adress{aoeu}}
\author{Victor Lagerkvist\corref{cor1}}
\cortext[cor1]{Corresponding author.}
\ead{victor.lagerkvist@liu.se}
%\fntext[label2]{}

\address{Department of Computer and Information Science, Link\"opings
  Universitet, Sweden.}
%\fntext[label3]{}
%\thanks{Partially supported by the Swedish Research Council (VR) under
%  grant 2008-4675.}}
%\institute{Department of Computer and Information Science, Link\"opings
%  Universitet, Sweden, 
%\email{\href{mailto:victor.lagerkvist@liu.se}{victor.lagerkvist@liu.se}}}

\date{}

%\maketitle

\begin{abstract}
Universal algebra and clone theory have proven to be a useful tool in
the study of constraint satisfaction problems since the complexity, up
to logspace reductions, is determined by the set of polymorphisms of
the constraint language. For classifications where primitive positive
definitions are unsuitable, such as size-preserving reductions, weaker
closure operations may be necessary. In this article we consider
strong partial clones which can be seen as a more fine-grained
framework than Post's lattice where each clone splits into an
interval of strong partial clones.  We investigate these intervals 
and give simple relational descriptions, weak
bases, of the largest elements. The weak bases have
a highly regular form and are in many cases easily relatable to the
smallest members in the intervals, which suggests that the lattice of
strong partial clones is considerably simpler than the full lattice of
partial clones.

%Weak bases are also have practical applications since
%each weak base give rise to the problem with the lowest complexity for
%all problems where the Galois connection can be applied a posteriori.

%% This lattice has a correlation with exact
%% complexity of constraint satisfaction problems in such a way that
%% constraint languages generating the same strong partial clone

%% Since some
%% of these intervals equal the continuum it is clear that the lattice of
%% partial clones is not nearly as well understood as the Post
%% lattice. 

%This in conjunction with the
%smallest members in these intervals, plain bases, suggests that
%the lattice of strong partial clones is much simpler than the lattice
%of partial clones.

\end{abstract}


\end{frontmatter}

\section{Introduction}
A set of functions is called a {\em clone} if it (1) is closed under
composition of functions and (2) contain all projection functions of
the form $e^n_i(x_1, \ldots, x_n) = x_i$. Dually, a set of relations
$\Gamma$ is called a relational clone, or a {\em co-clone}, if it
contains every relation $R$ definable through a {\em primitive
  positive} (p.p.) implementation of the form $R(x_1, \ldots, x_n)
\equiv \exists y_1, \ldots, y_m\, . \, R_1(\mathbf{x_1}) \wedge \ldots
R_k({\mathbf{x_k}})$, where each $R_i \in \Gamma \cup \{=\}$ and each
$\mathbf{x_i}$ is a vector over $x_1,\ldots, x_n$, $y_1, \ldots,
y_m$. In the case where $\Gamma$ is finite we say that it is a {\em
  constraint language}. For a set of functions $F$ and a set of
relations $\Gamma$ we use $[F]$ to denote the smallest clone
containing $F$ and $\cclone{\Gamma}$ for the smallest co-clone
containing $\Gamma$. If $\Gamma$ is a set of relations and
$\mathrm{I}\clone{C}$ a co-clone such that $\cclone{\Gamma} =
\mathrm{I}\clone{C}$ then we say that $\Gamma$ is a {\em base} of $\mathrm{I}\clone{C}$.
Ordering clones by set-inclusion yields a lattice
structure which in the Boolean case is completely explicated and known
as {\em Post's lattice} due to Post's seminal
classification~\cite{pos41}. Essentially the lattice determines the
expressive properties of all possible Boolean functions. Due to the
{\em Galois connection} between clones and co-clones the lattice of
Boolean co-clones is anti-isomorphic to Post's lattice and therefore
works as a complete classification of all Boolean languages. Simple
bases for all Boolean co-clones minimal with respect to arity of
relations have been identified by B\"ohler et al.$\,$\cite{bsrv05}
The lattice of Boolean co-clones is visualized in Figure
\ref{figure:postlattice}. The complexity of various computational
problems parameterized by constraint languages such as the {\em
  constraint satisfaction problem} (CSP) have been shown to be
determined up to logspace reducibility by Post's
lattice~\cite{allender2009,jeavons1998}. If one on the other hand is
interested in complexity classifications based on reductions which
preserves the exact complexity of problems, Post's lattice falls
short since even logspace reductions may introduce new
variables which affects the running-time.

\begin{figure}

\label{figure:postlattice}
\includegraphics[scale=0.36]{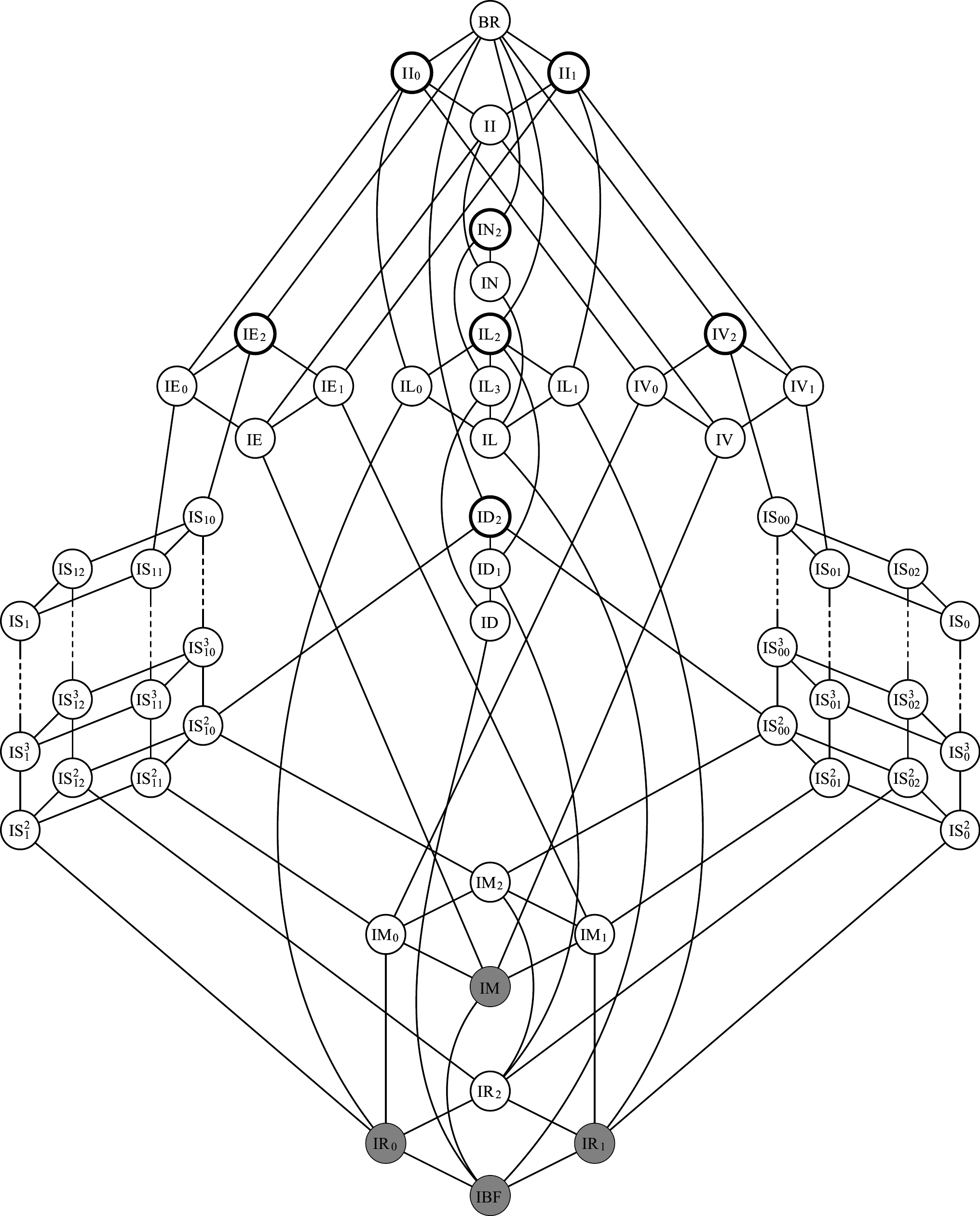}
\caption{The lattice of Boolean co-clones. The co-clones which are
  covered by a single weak partial co-clone are coloured in grey.}
\end{figure}

To remedy this a more fine-grained framework which further separates
constraint languages based on their expressive properties is
necessary. In Jonsson et al.$\,$\cite{Jonsson:etal:soda2013} the
lattice of {\em strong partial clones} is demonstrated to have the
required properties. Hence a classification of the lattice of strong
partial clones similar to that of Post's lattice would provide a
powerful framework for studying exact complexity of CSP and related
problems. We wish to emphasize that even though the lattice of partial
clones is known to be uncountable~\cite{alekseev1994} the same does
not necessarily hold for the lattice of strong partial
clones. Ideally, for each clone $\clone{C}$, one would like to
determine the interval of strong partial clones whose subset of total
functions equal $\clone{C}$. The strong partial clones in this
interval are said to {\em cover} $\clone{C}$. In Creignou et
al.$\,$\cite{creignou2008} relational descriptions known as {\em plain
  bases } of the smallest member of this interval is given. In this
article we give simple relational descriptions known as {\em weak
  bases} of the largest elements in these intervals. Our work builds
on the result of Schnoor and Schnoor~\cite{schnoor2008a,schnoor2008b}
but differs in two important aspects: first, each weak base presented
can in a natural sense be considered to be minimal; second, we present
alternative proofs where Schnoor's and Schnoor's procedure results in
relations which are exponentially larger than the bases given by
B\"ohler et al.$\,$\cite{bsrv05} and Creignou et
al.$\,$\cite{creignou2008}, and are thus also able to cover the
infinite chains in Post's lattice. Due to the Galois connection
between clones and co-clones the weak bases also constitutes the
relations which in a precise sense results in the CSP problems with
the lowest complexity~\cite{Jonsson:etal:soda2013}.  Hence the weak
bases presented in Section \ref{section:weak_bases} are closely
connected to upper bounds of running times for all problems
parameterized by constraint languages.

\section{Preliminaries}
\label{section:preliminaries}
%% In this section we present the Inv-Pol connection between clones and
%% co-clones, show how this can be extended to inaugurate partial clones,
%% and last present Schoor's and Schnoor's theorem for finding weak
%% bases. 
In this section we introduce some basic notions from universal algebra
and clone theory necessary for the construction of weak bases. If $f$
is an $n$-ary function and $R$ a relation with $m$ tuples it is possible
to extend $f$ to operate over tuples from $R$ as follows:
%\vspace{-10}
\begin{align*}
f(t_1,\dots ,t_n) = \big(& f(t_1[1], \dots, t_n[1]), \\ & \vdots \\
%& f(t_1[2], \dots, t_n[2]), \\ & \vdots \\
& f(t_1[m], \dots, t_n[m])\big),
\end{align*}
where $t_i[j]$ denotes the $j$-th argument of the tuple $t_i \in
R$. If $R$ is closed under $f$ we say that $f$ {\em preserves} $R$ or
that $f$ is a {\em polymorphism} of $R$. For a set of functions $F$ we
define Inv($F$) (often abbreviated as I$F$) to be the set of all
relations preserved by all functions in $F$. Dually we define
Pol($\Gamma$) for a set of relations $\Gamma$ to be the set of
polymorphisms to $\Gamma$. It is easy to verify that Pol($\Gamma$)
always form clones and that Inv($F$) always form co-clones. Moreover we have the
{\em Galois connection} between clones and co-clones
normally presented as: 

\begin{theorem}~\cite{BKKR69i,BKKR69ii,Gei68}
  Let $\Gamma$ and $\Delta$ be two sets of relations. Then
  $\cclone{\Gamma} \subseteq \cclone{\Delta}$ if and only if 
  Pol$(\Delta) \subseteq$ Pol$(\Gamma)$.

\end{theorem}

%% Ordering clones and co-clones by set inclusion yields a lattice
%% structure. In the Boolean case the lattice is completely classified
%% and known as the {\em Post lattice} due to Post's seminal
%% work~\cite{pos41}. See Figure \ref{figure:postlattice} for a
%% visualization of the co-clone lattice.
%For this reason more restricted closure operations than
%p.p. definitions are neccessary. 

To extend these notions to the case of partial clones we need some
additional notation. If $R$ is an $n$-ary Boolean relation and
$\Gamma$ a constraint language we say that $R$ has a {\em
  quantifier-free primitive positive} (q.p.p.) implementation in
$\Gamma$ if $R(x_1, \ldots, x_n) \equiv R_1(\mathbf{x_1}) \wedge
\ldots \wedge R_k({\mathbf{x_k}})$, where each $R_i \in \Gamma \cup \{=\}$
and each $\mathbf{ x_i}$ is a vector over $x_1,\ldots, x_n$. We use
$\pcclone{\Gamma}$ to denote the smallest set of relations closed
under q.p.p.\ definability. If I$\clone{C} =
\pcclone{\mathrm{I}\clone{C}}$ then we say that I$\clone{C}$ is a {\em
  weak partial co-clone}. We use the term weak partial co-clone to
avoid confusion with partial co-clones used in other contexts (see
Chapter 20.3 in Lau~\cite{lau2006}).
To get a corresponding concept on the functional
side we extend the previous definition of a polymorphism and say that
a partial function $f$ is a {\em partial polymorphism} to a relation
$R$ if $R$ is closed under $f$ for every sequence of tuples for which
$f$ is defined. A set of partial functions $\clone{C}$ is said to be a {\em partial
clone} if it contain all projection functions and is closed under
composition of functions. If $\clone{C}$ is a partial clone we say that it is
{\em strong} if for every $f \in \clone{C}$, $\clone{C}$ also contain
all partial subfunctions $g$ of $f$ which agrees with $f$ for all
values that they are defined.
%$g$ such that the domain of $g$ is a subset of the
%domain of $f$, and that $g$ agrees with $f$ for all values for which
%it is defined.
%We use $[F]_p$ to denote the smallest partial clone
%containing $F$ and $[F]_s$ to denote the smallest strong partial
%clone containing $F$. 
By pPol($\Gamma$) we denote the set of partial polymorphisms to the set of
relations $\Gamma$. Obviously sets of the form pPol($\Gamma$) always
form strong partial clones and again we have a Galois connection
between clones and co-clones.

\begin{theorem}~\cite{BKKR69i,BKKR69ii,romov1981}
  Let $\Gamma$ and $\Delta$ be two sets of relations. Then
  $\pcclone{\Gamma} \subseteq \pcclone{\Delta}$ if and only if 
  pPol$(\Delta) \subseteq$ pPol$(\Gamma)$.

\end{theorem}

%% Since minimal core-sizes are known for all Boolean co-clones
%% \cite{schnoor2008b} Theorem \ref{theorem:ccols_weak_base} can be used
%% to calculate weak bases for all Boolean co-clones with a finite
%% base. This procedure can however result in needlessly obfuscated and
%% complicated relations which are not always minimal. As an extreme case
%% the weak base S$^{n}_{0}$(COLS$_n$) for $\coclone{IS}{n}{0}$ has arity
%% $2^n$ and is not obviously relatable to the bases given by B\"ohler et
%% al\cite{bsrv05} or Creignou et al\cite{creignou2008}. 
%For any $n$-ary relation $R$ we let $R^{\F}$ be the $n+1$-ary relation $R \circ \F$. The
%relation $R^{\T}$ is defined dually.

%% For a clone $\clone{C}$ we define $\clone{I}(\clone{C}) = \{\clone{D} \mid
%% \clone{D}$ is a strong partial clone and $\clone{D} \cap
%% \mathrm{Pol}(\emptyset) = \clone{C}\}$. In other words
%% $\clone{I}(\clone{C})$ is the set of all strong partial clones
%% occurring inside of $\clone{C}$. We define $\clone{I}_{\cup}(\clone{C})$
%% to be the union of all clones in $\clone{I}(\clone{C})$.
For a co-clone I$\clone{C}$ we define $\clone{I}(\mathrm{I}\clone{C})
= \{\mathrm{I}\clone{D} \mid \mathrm{I}\clone{D} =
\pcclone{\mathrm{I}\clone{D}}$ and $\cclone{\mathrm{I}\clone{D}} =
\mathrm{I}\clone{C}\}$. In other words
$\clone{I}(\mathrm{I}\clone{C})$ is the interval of all weak partial
co-clones occurring inside of I$\clone{C}$. Let
$\clone{I}_{\cap}(\mathrm{I}\clone{C}) = \bigcap_{\mathrm{I}\clone{D} \in
  \clone{I}(\mathrm{I}\clone{C})}\mathrm{I}\clone{D}$. To be
consistent with Schnoor's and Schnoor's~\cite{schnoor2008a} notation
which is defined in terms of clones instead of co-clones we also
define $\clone{I}_{\cup}(\clone{C}) = \bigcup_{\mathrm{I}\clone{D} \in
  \clone{I}(\mathrm{I}\clone{C})}\mathrm{pPol}(\mathrm{I}\clone{D})$. Obviously
$\clone{I}_{\cup}(\clone{C})$ is the union of all strong partial
clones covering $\clone{C}$, from which it follows that
pPol($\clone{I}_{\cap}(\mathrm{I}\clone{C})) = \clone{I}_{\cup}(\clone{C})$.

\begin{definition}
  Let $\clone{C}$ be a clone. A constraint language $\Gamma$ is a {\em
    weak base} of I$\clone{C}$ if pPol($\Gamma$) $=$
  $\clone{I}_{\cup}(\clone{C})$.
\end{definition}

%A weak base for a co-clone I$\clone{C}$ therefore has the largest
%number of partial polymorphisms of the partial co-clones in the
%interval of I$\clone{C}$ and as a consequence it is the smallest
%partial co-clone inside a given co-clone.
Due to the Galois connection between strong partial clones and weak partial
co-clones a weak base for a co-clone I$\clone{C}$ therefore results in
smallest element in $\clone{I}(\mathrm{I}\clone{C})$. The
following theorem is immediate from the definition and the fact
that pPol($\clone{I}_{\cap}(\mathrm{I}\clone{C})) =
\clone{I}_{\cup}(\clone{C})$.

\begin{theorem} [\cite{schnoor2008a}] \label{theorem:weak_base}
  Let $\clone{C}$ be a clone and $\Gamma$ be a weak base of
  I$\clone{C}$. Then, for any base $\Gamma'$ of I$\clone{C}$,
  it holds that $\Gamma \subseteq \pcclone{\Gamma'}$.
\end{theorem}

%% \begin{theorem} \label{theorem:weak_base2}
%%   Let $\clone{C}$ be a clone. $\Gamma$ is a weak base of
%%   I$\clone{C}$ if and only if $\Gamma \subseteq \pcclone{\Gamma'}$
%%   for any base $\Gamma'$ of I$\clone{C}$.
%% \end{theorem}

%% \begin{proof}
%%   The first direction follows immediately from Theorem
%%   \ref{theorem:weak_base}. For the other direction we assume that
%%   $\Gamma \subseteq \pcclone{\Gamma'}$. By ... there is at least one
%%   weak base $\Delta$ of Inv($\clone{C}$) from which it follows that
%%   $\Delta \subseteq \pcclone{\Gamma}$. But by assumption it also holds
%%   that $\Gamma \subseteq \pcclone{\Delta}$ and hence that
%%   $\pcclone{\Gamma} = \pcclone{\Delta}$ --- therefore pPol($\Gamma$) $=$
%%   $\clone{I}_{\cup}(\clone{C})$ and $\Gamma$ is a weak base of
%%   Inv$(\clone{C})$. \qed
%%  \end{proof}

If $R$ is an $n$-ary relation with $m = |R|$ elements we let the
matrix representation of $R$ be the $m \times n$-matrix containing the
tuples of $R$ as rows stored in lexicographical order. Note that the
ordering is only relevant to ensure that the representation is
unambiguous. Given a natural number $n$ the $2^n$-ary relation
COLS$^n$ is the relation which contains all natural numbers from $0$
to $2^n - 1$ as columns in the matrix representation. 
%In other words each
%column 
%\[c_i = 
%\begin{pmatrix} 
%  x_{i,1} \\
%  \vdots \\
%  x_{i,n}
%\end{pmatrix}\] 
%of COLS$^n$ is a binary representation of $i$ such
%that $x_{i,1} \hdots x_{i,n} = i$. 
For any clone $\clone{C}$ and relation $R$ we define $\clone{C}(R)$ to
be the relation $\bigcap_{R' \in \mathrm{I}\clone{C}, R \subseteq R'}
R'$, i.e. the smallest extension of $R$ which is preserved under every
function in $\clone{C}$. For a relation $R$ we say that the co-clone
$\cclone{R}$ has {\em core-size} $s$ if there is a relation $R'$ such
that $|R'| = s$ and $R = (\mathrm{Pol}(R))(R')$. Minimal core-sizes
for all Boolean co-clones have been identified by
Schnoor~\cite{schnoor2008b}. We are now ready to state Schnoor's and
Schnoor's~\cite{schnoor2008a} main result which effectively gives a
weak base for any co-clone with a finite core-size.

%% If $R$ is a relation such that $\cclone{R} = $ Inv($\clone{C})$ we say
%% that Inv($\clone{C}$) has {\em core-size} $s$ if there is a relation
%% $R'$ with cardinality $s$ such that $R = \clone{C}(R')$.

\begin{theorem}[\cite{schnoor2008a}] \label{theorem:ccols_weak_base}
  Let $\clone{C}$ be a clone and $s$ be a core-size of
  I$\clone{C}$. Then the relation $\clone{C}$(COLS$^s$) is a weak
  base of I$\clone{C}$.
\end{theorem}

The disadvantage of the theorem is that relations of the form
$\clone{C}$(COLS$^s$) have exponential arity with respect to the
core-size. We therefore introduce another measurement of minimality
which ensures that a given relation is indeed minimal with respect to
cardinality. A relation $R$ is said to be {\em irredundant} if there
are no duplicate rows in the matrix representation.
%TODO: Also: no fictious/redundant variables.

\begin{definition}
  A relation $R$ is {\em minimal} if it is irredundant and there is no
  $R' \subset R$ such that $\cclone{R} = \cclone{R'}$.
\end{definition}

Minimal weak bases have the property that they can be implemented
without the use of the equality operator. If we let $\langle \cdot
\rangle_{\not \exists \not =}$ denote the closure of
q.p.p.\ definitions without equality we therefore get the following
theorem.

\begin{theorem}[\cite{schnoor2008a}] \label{theorem:weak_base3}
  Let $\clone{C}$ be a clone and $\Gamma$ be a minimal weak base of
  I$\clone{C}$. Then, for any base $\Gamma'$ of I$\clone{C}$,
  it holds that $\Gamma \subseteq \langle \Gamma' \rangle_{\not
    \exists \not =}$.
\end{theorem}

Hence minimal weak bases give the largest possible expressibility
results and are applicable for problems where the equality operator is
not permissable, e.g. counting CSP, where the number of solutions can
be increased by an exponential factor~\cite{schnoor2008a}.

\section{Minimal weak bases of all Boolean co-clones}
\label{section:weak_bases}

%% If $R = \{t_1,\ldots, t_m\}$ is an $n$-ary relation elements we let the
%% matrix representation of $R$ be the $n \times m$-matrix \[R
%% = \begin{pmatrix} 
%%   t_{\pi(1)} \\
%%   \vdots \\
%%   t_{\pi(n)}
%% \end{pmatrix}\]
%% where $\pi$ is a permutation satisfying .

In this section we proceed by giving minimal weak bases for all Boolean
co-clones with finite core-size. The results are presented in Table
\ref{table:weak_bases}. Each line in the table consists of a co-clone,
its minimal core-size and a minimal weak base. As convention we use
normal Boolean connectives to represent relations whenever this
promotes readability. For example $x_1x_2$ denotes the relation $\{(1,1)\}$
while $x_1 \neq x_2$ denotes the relation $\{(0,1),(1,0)\}$. We
use $\F$ for the relation $\{(0)\}$ and $\T$ for the relation
$\{(1)\}$. The relations $\orn{}{n}$ and $\nandn{}{n}$ are $n$-ary
{\em or}
and {\em nand}. $\evenn{}{n}$ is the $n$-ary relation which holds if the sum of its
arguments is even, and conversely for $\oddn{}{n}$. By $\R$ we denote
the 3-ary relation $\{(0,0,1), (0,1,0), (1,0,0)\}$. If $R$ is an
$n$-ary relation we often use $R_{m \neq}$ to denote the $(n+m)$-ary
relation defined as $R_{m \neq}(x_1, \ldots, x_{n + m}) \equiv R(x_1,
\ldots, x_n) \wedge (x_1 \neq x_{n+1}) \wedge \ldots \wedge (x_n \neq
x_{n + m})$. Variables are named $x_1, \ldots, x_n$ or $x$ except when
they occur in $\F$ or $\T$ in which case they are named $c_0$ and
$c_1$ respectively to explicate that they are in essence constant
values.

\begin{table*}  \scriptsize
\caption{Weak bases for all Boolean co-clones with a finite base}
\label{table:weak_bases}
\begin{tabularx}{\textwidth}{l c l c c}
  \hline
  Co-clone & Core-size & Weak base \\
  \hline
  $\coclone{IBF}{}{}$ & 1 & Eq$(x_1,x_2)$ \\
  $\coclone{IR}{}{0}$ & 1 & $\F(c_0)$  \\
  $\coclone{IR}{}{1}$ & 1 & $\T(c_1)$  \\
  $\coclone{IR}{}{2}$ & 1 & $\F(c_0) \wedge \T(c_1)$  \\

  $\coclone{IM}{}{}$ & 1 & $(x_1 \rightarrow x_2)$  \\
  $\coclone{IM}{}{0}$ & 2 & $(x_1 \rightarrow x_2) \wedge \F(c_0)$ \\ 

  $\coclone{IM}{}{1}$ & 2 & $(x_1 \rightarrow x_2) \wedge \T(c_1)$ \\

  $\coclone{IM}{}{2}$ & 3 & $(x_1 \rightarrow x_2) \wedge \F(c_0)
  \wedge \T(c_1)$  \\

  $\coclone{IS}{n}{0}, n \geq 2 $ & $n$ & $\orn{}{n}(x_1, \ldots, x_n) \wedge \T(c_1)$  \\
 % $\coclone{IS}{}{0} $ & $\infty$   \\

%  $\coclone{IS}{}{1} $ & $\infty$   \\

  $\coclone{IS}{n}{02}, n \geq 2$ & $n$ & $\orn{}{n}(x_1, \ldots, x_n) \wedge
  \F(c_0) \wedge \T(c_1)$ \\
  $\coclone{IS}{n}{01}, n \geq 2$ & $n$ & $\orn{}{n}(x_1, \ldots, x_n)
  \wedge (x \rightarrow x_1 \cdots x_n)
   \wedge \T(c_1)$  \\
  $\coclone{IS}{n}{00}, n \geq 2$ & $\max(3, n)$ & $\orn{}{n}(x_1, \ldots, x_n)
  \wedge (x \rightarrow x_1 \cdots x_n)
    \wedge \F(c_0) \wedge \T(c_1)$  \\
%  $\coclone{IS}{}{02} $ & $\infty$   \\
  $\coclone{IS}{n}{1}, n \geq 2$ & $n$ & $\nandn{}{n}(x_1, \ldots, x_n) \wedge
  \F(c_0)$  \\
  $\coclone{IS}{n}{12}, n \geq 2 $ & $n$ & $\nandn{}{n}(x_1, \ldots, x_n) \wedge
  \F(c_0) \wedge \T(c_1)$  \\
%  $\coclone{IS}{}{12} $ & $\infty$   \\
%  $\coclone{IS}{}{01} $ & $\infty$   \\
  $\coclone{IS}{n}{11}, n \geq 2$ & $n$ & $\nandn{}{n}(x_1, \ldots, x_n)
  \wedge (x \rightarrow x_1 \cdots x_n)
   \wedge \F(c_0)$  \\
%  $\coclone{IS}{}{11} $ & $\infty$   \\

%  $\coclone{IS}{n}{01} $    \\
%  $\coclone{IS}{}{01} $ & $\infty$   \\
%  $\coclone{IS}{n}{11} $    \\
%  $\coclone{IS}{}{11} $ & $\infty$   \\

%  $\coclone{IS}{}{00} $ & $\infty$   \\
  $\coclone{IS}{n}{10}, n \geq 2$ & $\max(3, n)$ & $\nandn{}{n}(x_1, \ldots, x_n)
  \wedge (x \rightarrow x_1 \cdots x_n)
   \wedge \F(c_0) \wedge \T(c_1)$  \\
%  $\coclone{IS}{}{10} $ & $\infty$   \\

  $\coclone{ID}{}{}$ & 1 & $(x_1 \neq x_2)$ \\
  $\coclone{ID}{}{1}$ & 2 & $(x_1 \neq x_2) \wedge \F(c_0) \wedge \T(c_1)$  \\
  $\coclone{ID}{}{2}$ & 3 & $\orn{2 \neq}{2}(x_1,x_2,x_3,x_4) \wedge \F(c_0) \wedge \T(c_1)$  \\

  $\coclone{IL}{}{}$  & 2 & $\evenn{}{4}(x_1,x_2,x_3,x_4)$ \\
  $\coclone{IL}{}{0}$ & 2 & $\evenn{}{3}(x_1,x_2,x_3) \wedge \F(c_0)$  \\
  $\coclone{IL}{}{1}$ & 2 & $\oddn{}{3}(x_1,x_2,x_3) \wedge \T(c_1)$  \\
  $\coclone{IL}{}{2}$ & 3 & $\evenn{3 \neq}{3}(x_1,\ldots,x_6) \wedge \F(c_0) \wedge \T(c_1)$  \\
  $\coclone{IL}{}{3}$ & 3 & $\evenn{4 \neq}{4}(x_1,\ldots,x_8)$  \\

  $\coclone{IV}{}{}$ & 2 & $(\overbar{x_1} \leftrightarrow
\overbar{x_2}\overbar{x_3}) \wedge (\overbar{x_2} \vee \overbar{x_3}
\rightarrow \overbar{x_4})$ \\
  $\coclone{IV}{}{0}$ & 2 & $(\overbar{x_1} \leftrightarrow \overbar{x_2}\overbar{x_3}) \wedge \F(c_0)$  \\
  $\coclone{IV}{}{1}$ & 3 & $(\overbar{x_1} \leftrightarrow
\overbar{x_2}\overbar{x_3}) \wedge (\overbar{x_2} \vee \overbar{x_3}
\rightarrow \overbar{x_4}) \wedge \T(c_1)$  \\
  $\coclone{IV}{}{2}$ & 3 & $(\overbar{x_1} \leftrightarrow
\overbar{x_2}\overbar{x_3}) \wedge \F(c_0) \wedge \T(c_1)$  \\

  $\coclone{IE}{}{}$ & 2 & $(x_1 \leftrightarrow x_2x_3) \wedge (x_2 \vee x_3
\rightarrow x_4)$ \\
  $\coclone{IE}{}{0}$ & 3 & $(x_1 \leftrightarrow x_2x_3) \wedge (x_2 \vee x_3
\rightarrow x_4) \wedge \F(c_0)$  \\
  $\coclone{IE}{}{1}$ & 2 & $(x_1 \leftrightarrow x_2x_3) \wedge \T(c_1)$  \\
  $\coclone{IE}{}{2}$ & 3 & $(x_1 \leftrightarrow x_2x_3) \wedge \F(c_0) \wedge \T(c_1)$  \\

  $\coclone{IN}{}{}$ & 2 & $\evenn{}{4}(x_1,x_2,x_3,x_4) \wedge x_1x_4
\leftrightarrow x_2x_3$ \\
  $\coclone{IN}{}{2}$ & 3 & $\evenn{4 \neq}{4}(x_1,\ldots,x_8) \wedge x_1x_4
\leftrightarrow x_2x_3$  \\  

  $\coclone{II}{}{}$ & 2 & $(x_1 \leftrightarrow x_2x_3) \wedge (\overbar{x_4}
\leftrightarrow \overbar{x_2}\overbar{x_3})$ \\
  $\coclone{II}{}{0}$ & 2 & $(\overbar{x_1} \vee \overbar{x_2}) \wedge (\overbar{x_1}\overbar{x_2} \leftrightarrow \overbar{x_3}) \wedge \F(c_0)$  \\
  $\coclone{II}{}{1}$ & 2 & $(x_1 \vee x_2) \wedge (x_1x_2 \leftrightarrow x_3) \wedge \T(c_1)$  \\
  $\coclone{BR}{}{}$ & 3 & $\Rddd(x_1,\ldots,x_6) \wedge \F(c_0)
\wedge \T(c_1)$ \\

  \hline
\end{tabularx}
\end{table*}

For the co-clones $\coclone{IR}{}{2}$, $\coclone{IM}{}{}$,
$\coclone{ID}{}{}$, $\coclone{ID}{}{1}$, $\coclone{IL}{}{}$,
$\coclone{IL}{}{0}$, $\coclone{IL}{}{1}$, $\coclone{IL}{}{2}$,
$\coclone{IL}{}{3}$, $\coclone{IV}{}{}$, $\coclone{IV}{}{0}$,
$\coclone{IE}{}{}$, $\coclone{IE}{}{1}$, $\coclone{IN}{}{}$,
$\coclone{IN}{}{2}$, $\coclone{II}{}{}$, $\coclone{II}{}{0}$,
$\coclone{II}{}{1}$ and $\coclone{BR}{}{}$, the result follows
immediately from Theorem \ref{theorem:ccols_weak_base}, the minimal
core-sizes for each co-clone, and a suitable rearrangement of
arguments. Through exhaustive search, i.e. by repeatedly removing
redundant columns and tuples, one can verify that the bases are
also minimal. This has been done by a computer program which is
available upon request from the author. For the remaining co-clones
the proof is divided into two parts. First, we prove that the weak
base for every co-clone I$\clone{C}$ in $\coclone{IM}{}{0}$,
$\coclone{IM}{}{1}$, $\coclone{IM}{}{2}$, $\coclone{ID}{}{2}$,
$\coclone{IV}{}{1}$, $\coclone{IV}{}{2}$, $\coclone{IE}{}{0}$ and
$\coclone{IE}{}{2}$, can be obtained by collapsing columns from
$\clone{C}$(COLS$^s$). Second, we prove that for every $n \geq 2$
there exists simple weak bases for the co-clones $\coclone{IS}{n}{0}$,
$\coclone{IS}{n}{02}$, $\coclone{IS}{n}{01}$, $\coclone{IS}{n}{00}$
and their duals $\coclone{IS}{n}{1}$, $\coclone{IS}{n}{12}$,
$\coclone{IS}{n}{11}$, $\coclone{IS}{n}{10}$.

To make the proofs more concise we introduce some admissible operations
on relations which preserves the weak base property. Let $R$ be an
$n$-ary relation. Each rule is of the form $R \deducesto{} R'$ and
implies that $\pcclone{R'} \subseteq \pcclone{R}$.

\begin{itemize}
\item $R \deducesto{(i = j)} R'$, $1 \leq i < j \leq n$,
\\ (Identify argument $i$ with argument $j$),
\item
  $R \deducesto{\pi(i_1, \ldots, i_n)} R'$, where $\pi$ is the
  permutation $\pi(j) = i_j, 1 \leq
  j \leq n$, $1 \leq i_j \leq n$, \\
  (Swap arguments),
\item
  $R \deducesto{\mathrm{irr}} R'$, 
  \\  ($R'$ is the irredundant core of $R$).
\end{itemize}

%We
%write $R \deducesto{(i,j)} R'$ to denote that $R'$ is the relation
%$R'(x_1, \ldots, x_i, \ldots, x_n) \equiv $

\begin{lemma}
  Let I$\clone{C}$ be a co-clone, $R$ an $n$-ary weak base for I$\clone{C}$, and let $R'$ be a Boolean
  relation such that $R \deducesto{*} R'$ for some rule $\deducesto{*}$. If $R'$ is a base
  for I$\clone{C}$ then it is also a weak base for I$\clone{C}$.

%If $\cclone{R} =
 % \cclone{R'}$ then $\pcclone{R} = \pcclone{R'}$.
\end{lemma}

\begin{proof}
  We prove that $\pcclone{R} = \pcclone{R'}$ which implies that
  $\clone{I}_{\cup}(\mathrm{Pol}(R)) =
  \clone{I}_{\cup}(\mathrm{Pol}(R'))$ and that $R'$ is a weak
  base for I$\clone{C}$. The first inclusion $\pcclone{R} \subseteq \pcclone{R'}$ is
  obvious since $R$ is a weak base by assumption. To prove that
  $\pcclone{R'} \subseteq \pcclone{R}$ we show that $R' \in
  \pcclone{R}$ by giving a q.p.p.\ implementation of $R'$ with
  $R$. There are two cases to consider. Either $R \deducesto{(i = j)}
  R'$, $1 \leq i < j \leq n$, in which case $R'$ is the ($n-1$)-ary
  relation defined as $R'(x_1, \ldots, x_i, \ldots,
  x_{n-1}) \equiv R(x_1, \ldots, \underbrace{x_i, \ldots,
    x_i}_{\mbox{\footnotesize j - i + 1}}, \ldots, x_{n-1})$, or $R
  \deducesto{\pi(i_1, \ldots, i_n)} R'$, in which case $R'(x_1, \ldots,
  x_n) \equiv R(x_{\pi(1)}, \ldots, x_{\pi(n)})$. The case when $R'$
  is the irredundant core of $R$ follows trivially from this since
  $R'$ can be obtained by identifying all variables that are equal.
%\underbrace{c_1,\ldots, c_1}_{\mbox{\footnotesize $m \:$ $c_1$'s}})
\end{proof}

\begin{lemma}
  The bases for $\coclone{IM}{}{0}$,
  $\coclone{IM}{}{1}$, $\coclone{IM}{}{2}$, $\coclone{ID}{}{2}$,
  $\coclone{IV}{}{1}$, $\coclone{IV}{}{2}$, $\coclone{IE}{}{0}$ and
  $\coclone{IE}{}{2}$ in Table \ref{table:weak_bases} are minimal weak bases.
\end{lemma}

\begin{proof} 
  We consider each case in turn. For every co-clone I$\clone{C}$ we
  write $R_{\mathrm{I}\clone{C}}$ for the weak base from Table
  \ref{table:weak_bases}, and $R$, $R'$, $\ldots$, for intermediate
  relations in the derivation.
  \\
  \noindent
  $\coclone{IR}{}{0}$: $\coclone{R}{}{0}(\mathrm{COLS}^1) \deducesto{(1 = 2)} R \deducesto{\mathrm{irr}} R_{\coclone{IR}{}{0}}$.
  \\
  \noindent
  $\coclone{IR}{}{1}$: $\coclone{R}{}{1} (\mathrm{COLS}^1)
  \deducesto{(1 = 2)} R \deducesto{\mathrm{irr}} R_{\coclone{IR}{}{1}}$.
  \\
  \noindent
  $\coclone{IM}{}{0}$: $\coclone{M}{}{0}(\mathrm{COLS}^2) \deducesto{(1 = 2)} R \deducesto{\mathrm{irr}} R' \deducesto{\pi(3, 1, 2)} R_{\coclone{IM}{}{0}}$.
  \\
  $\coclone{IM}{}{1}$: $\coclone{M}{}{1}(\mathrm{COLS}^2)
  \deducesto{(1 = 2)} R \deducesto{\mathrm{irr}} R_{\coclone{IM}{}{1}}$.
  \\
  $\coclone{IM}{}{2}$: $\coclone{M}{}{2}(\mathrm{COLS}^3)
  \deducesto{(1 = 2)} R \deducesto{(1 = 2)} R' \deducesto{(2 =
    3)} R'' \deducesto{\mathrm{irr}}
  R''' \deducesto{\pi(3, 1, 2, 4)} R_{\coclone{IM}{}{2}}$.
  \\
  $\coclone{ID}{}{2}$: $\coclone{D}{}{2}(\mathrm{COLS}^3). 
  \deducesto{(1 = 2)} R \deducesto{\mathrm{irr}} R' \deducesto{\pi(5, 4, 1, 3, 2, 6)} R_{\coclone{ID}{}{2}}$.
  \\
  $\coclone{IV}{}{1}$: $\coclone{V}{}{1}(\mathrm{COLS}^3)
  \deducesto{(4 = 8)} R \deducesto{(2 = 4)} R' \deducesto{(3 = 6)}
  R'' \deducesto{\mathrm{irr}} R'''
  \deducesto{\pi(4,2,3,1,5)} R_{\coclone{IV}{}{1}}$.
  \\
  $\coclone{IV}{}{2}$: $\coclone{V}{}{2}(\mathrm{COLS}^3) 
  \deducesto{(4 = 8)} R \deducesto{(2 = 4)} R' \deducesto{(3 = 6)}
  R'' \deducesto{\mathrm{irr}} R'''
  \deducesto{\pi(4,2,3,1,5)} R_{\coclone{IV}{}{2}}$.
  \\
  $\coclone{IE}{}{0}$: $\coclone{E}{}{0}(\mathrm{COLS}^3)
  \deducesto{(1 = 2)} R \deducesto{(1 = 2)} R' \deducesto{(1 = 2)}
  R'' \deducesto{\mathrm{irr}} R'''
  \deducesto{\pi(5,1,2,3,4)} R_{\coclone{IE}{}{0}}$.
  \\
  $\coclone{IE}{}{2}$: $\coclone{E}{}{2}(\mathrm{COLS}^3)
  \deducesto{(1 = 2)} R \deducesto{(1 = 2)} R' \deducesto{(1 = 2)}
  R'' \deducesto{\mathrm{irr}} R''' \deducesto{\pi(4,1,2,3,5)}
  R_{\coclone{IE}{}{2}}$.
  \\
  \\
  \noindent
  It is not hard to see that every relation $R_{\mathrm{I}\clone{C}}$ is
  a base of $\mathrm{I}\clone{C}$. As in the previous cases the
  minimality of each weak base can be verified through exhaustive
  search. As an example consider
  \[R_{\coclone{IE}{}{2}} = 
  \begin{pmatrix} 
    %% 0 & 0 & 0 & 0 & 1 \\
    %% 0 & 0 & 0 & 1 & 1 \\
    %% 0 & 0 & 1 & 0 & 1 \\
    %% 0 & 1 & 1 & 1 & 1
     0 & 0 & 0 & 0 & 1 \\
     0 & 0 & 1 & 0 & 1 \\
     0 & 1 & 0 & 0 & 1 \\
     1 & 1 & 1 & 0 & 1
  \end{pmatrix}. \] Removing three rows results in a relation in
  $\coclone{IR}{}{2}$ while removing two rows from
  $R_{\coclone{IE}{}{2}}$ results in a relation in
  $\coclone{ID}{}{1}$. Removing the first row results in a relation
  which generates $\coclone{BR}{}{}$ and is hence no longer included
  in $\coclone{IE}{}{2}$, removing the second or third row gives a
  relation in $\coclone{IM}{}{2}$, and removing the fourth row gives a
  relation in $\coclone{IS}{2}{10}$. Hence there is no relation $R'
  \subset R_{\coclone{IE}{}{2}}$ such that $\cclone{R'} =
  \coclone{IE}{}{2}$ by which it follows that $R_{\coclone{IE}{}{2}}$
  is a minimal weak base. 
 \end{proof}

We now turn our attention towards the infinite parts of Post's
lattice. In the sequel we sometimes represent relations by formulas in
conjunctive normal form. If $\mathbf{x} = x_1, \ldots, x_n$ we use
$\phi(\mathbf{x})$ to denote a formula with $n$ free variables. 
%The
%notion of prime implicates will also be important in the following
%proofs. 
If $\phi = C_1 \wedge \ldots \wedge C_m$ is a formula with $m$
clauses we say that $C_i$ is a {\em prime implicate} of $\phi$ if
$\phi$ does not entail any proper subclause of $C_i$. A formula $\phi$
is said to be {\em prime} if all of its clauses are prime
implicates. Obviously any finite Boolean relation is
representable by a prime formula. If $R$ is an $n$-ary Boolean
relation we can therefore prove that $R \in \langle \Gamma
\rangle_{\not \exists}$ by showing that $R(x_1, \ldots, x_n)$ can be
expressed as a conjunction $\phi_1(\mathbf{y_1}) \wedge \ldots \wedge
\phi_k(\mathbf{y_k})$, where each $\mathbf{y_i}$ is a vector over
$x_1, \ldots, x_n$ and each $\phi_i$ is a prime formula representation
of a relation in $\Gamma$. This is advantageous since relations in
$\coclone{IS}{n}{0}$, $\coclone{IS}{n}{02}$, $\coclone{IS}{n}{01}$,
$\coclone{IS}{n}{00}$, $\coclone{IS}{n}{1}$, $\coclone{IS}{n}{12}$,
$\coclone{IS}{n}{11}$ and $\coclone{IS}{n}{10}$ are representable by
prime {\em implicative hitting set-bounded} (IHSB)
formulas~\cite{creignou2008}. We let IHSB$+^n$ be the set of formulas
of the form $(x_1 \vee \ldots \vee x_m), 1 \leq m \leq n, (\neg x_1), (\neg x_1
\vee x_2)$, and dually for IHSB$-^n$. To avoid repetition we only present the
full proof for $\coclone{IS}{n}{00}$. The proofs for the other cases
follow through similar arguments.

\begin{lemma} \label{theorem:is00}
  The relation $R_{\coclone{IS}{n}{00}}(x_1, \ldots, x_n, x, c_0, c_1)$ $\equiv \Or(x_1, \ldots, x_n)
  \wedge (x \rightarrow x_1\cdots x_n)
  \wedge \F(c_0) \wedge \T(c_1)$ is a minimal weak base for $\coclone{IS}{n}{00}$.
\end{lemma}

\begin{proof}
    Let $\Gamma$ be a constraint language such that $\langle \Gamma
    \rangle = \coclone{IS}{n}{00}$. Since $\Gamma$ is finite we can
    without loss of generality restrict the proof to a single relation $R$ defined to be the
    cartesian product of all relations in $\Gamma$. We must prove that
    $R_{\coclone{IS}{n}{00}} \in \langle R \rangle_{\not \exists}$. By
    Creignou et al.$\,$\cite{creignou2008} we know that R can be
    expressed as an IHSB$+^n$ formula $\phi(y_1,\ldots, y_m)$.

    We first implement $\F(c_0)$ by identifying every
    variable $y_i$ occurring in a negative clause $(\neg y_i)$ to
    $c_0$. Note that there must exist at least one negative unary
    clause since otherwise $\langle R \rangle =
    \coclone{IS}{n}{01}$. Then, for any implicative clause $(\neg y_i
    \vee c_0)$ which also entails $(\neg c_0 \vee y_i)$ we identify $y_i$
    with $c_0$. For any remaining clause we identify all unbound
    variable with $c_1$. Since there must exist at least one positive
    prime clause this correctly implements $\T(c_1)$.

    Since $\langle R \rangle = \coclone{IS}{n}{00}$ there is at least
    one $n$-ary prime clause of the form $(y_1 \vee \ldots \vee y_n)$
    in $\phi$. We can therefore implement $\Or(x_1, \ldots, x_n)$ with
    $\phi(y_1,\ldots,y_m)$ by first identifying $y_1, \ldots, y_n$ and
    $x_1, \ldots, x_n$. Let the resulting formula be $\phi'$. Note
    that $\phi'$ might still contain unbound variables. In the
    subsequent formula we use $x_i$, $1 \leq i \leq n$, to denote a
    variable in $x_1, \ldots, x_n$ and $y_j$, $n + 1 \leq j \leq m$,
    to denote a variable in $y_{n+1}, \ldots, y_m$. Hence we need to
    replace each $y_j$ still occurring in $\phi'$ with $x_i$, $c_0$,$c_1$ or
    $x$. For every implicative clause $C$ in $\phi'$ there are four
    cases to consider:

    %In the subsequent formula we use $y_i$ to denote a
    %variable distinct from $x_1, \ldots, x_n$. Hence we need to
    %replace each $y_i$ still occurring in $\phi'$ with $x_i$, $c_0$,$c_1$ or
    %$x$. 
  
  \begin{enumerate}
  \item
    $C = (\neg x_i \vee x_i')$
  \item
    $C = (\neg x_i \vee y_j)$
  \item
    $C = (\neg y_j \vee y_j')$
  \item
    $C = (\neg y_j \vee x_i)$
  \end{enumerate}

  The first case is impossible since $(x_1 \vee \ldots \vee x_n)$ was
  assumed to be prime. This also implies that the clauses $(\neg x_i
  \vee y_j)$ and $(\neg y_j \vee x_i')$ cannot occur simultaneously in
  the formula. For the second case we identify $y_j$ with $c_1$. For
  the third case we identify both $y_j$ and $y_j'$ with $x$. For the
  fourth case we identify $y_j$ with $c_0$. As can be verified the
  resulting formula implements $\Or(x_1, \ldots, x_n)$. 

  In order to implement $(\neg x \vee x_1 \wedge \ldots \wedge x_n)$ we
  need to ensure that $\neg x \vee x_i$ for all $1 \leq i \leq
  n$. Since $\langle R \rangle = \coclone{IS}{n}{00}$ its prime
  formula representation $\phi$ must contain a prime clause of the
  form $(\neg y_j \vee y_j')$ where $\phi$ does not entail $(\neg y_j' \vee
  y_j)$. To implement $(\neg x \vee x_i)$ we therefore identify $y_j$
  with $x$ and $y_j'$ with $x_i$. In the subsequent formula there are
  three implicative cases to consider:

  \begin{enumerate}
  \item
    $C = (\neg x \vee y_j)$
  \item
    $C = (\neg x_i \vee y_j)$
  \item
    $C = (\neg y_j \vee x_i)$
  \end{enumerate}

  In the first case we identify $y_j$ with $x_i$, in the second case
  we identify $y_j$ with $c_1$, and in the third case we identify
  $y_j$ with $x$. For any remaining positive clause we identify each
  unbound variable to $c_1$, and for any remaining negative unary
  clause $(\neg y_j)$ we identify $y_j$ with $c_0$. If we repeat the
  procedure for all $1 \leq i \leq n$ we see that $(\neg x \vee x_1
  \wedge \ldots \wedge x_n)$. All resulting formulas now only contain
  variables from $x_1, \ldots, x_n$, $x$, $c_0$, $c_1$, and hence the
  implementation is indeed a q.p.p.\ implementation. 

  One can also prove that $R_{\coclone{IS}{n}{00}}$ is a base of
  $\coclone{IS}{n}{00}$ by giving an explicit p.p.\ definition of the
  base given by B\"ohler et al.\cite{bsrv05}.  As for
  the minimality we simply note that removing any tuple from
  $R_{\coclone{IS}{n}{00}}$ results in a relation which is no longer
  a base of $\coclone{IS}{n}{00}$.
\end{proof}

Due to the duality of $\coclone{IS}{n}{0}, \coclone{IS}{n}{02},
\coclone{IS}{n}{01}, \coclone{IS}{n}{00}$ with $\coclone{IS}{n}{1},
\coclone{IS}{n}{12}, \coclone{IS}{n}{11}, \coclone{IS}{n}{10}$ we
skip the latter proofs and instead refer to Lemma
\ref{theorem:is00}. We have thus proved the main result of the paper.

\begin{theorem}
  The relations in Table \ref{table:weak_bases} are minimal weak bases.
\end{theorem}

\section{Conclusions and future work}
We have determined minimal weak bases for all Boolean co-clones with a
finite base. Below are some topics relevant for future pursuits.

\noindent
{\bf The lattice of strong partial clones.} Since the weak
and plain base of a co-clone I$C$ constitute the smallest and largest
weak partial co-clone occurring inside of I$\clone{C}$ it would be
interesting to determine the full interval of weak partial co-clones
between the weak base and the plain base. Especially one would like
to determine whether these intervals are finite, countably infinite or
equal to the continuum.
%Is this interval finite or
%countably infinite for all Boolean co-clones or does it equal the
%continuum?

\noindent
{\bf Exact complexity of constraint problems.} Each weak base
effectively determines the constraint problem with the lowest
complexity in a given co-clone. Example applications which follows
from the categorization in this article include the easiest NP-complete Boolean
CSP($\cdot$) problem in Jonsson et al.$\,$\cite{Jonsson:etal:soda2013}
which is simply the weak base of $\coclone{BR}{}{}$ without
constant columns. Are there other problems besides Boolean CSP($\cdot$) which admits a
single easiest problem?

\section*{Acknowledgements}
The author is grateful towards Peter Jonsson, Gustav Nordh and Bruno
Zanuttini for helpful comments and suggestions.

\bibliography{references}
\bibliographystyle{plain}
\end{document}